\documentclass{ar}

\usepackage[comma]{natbib}
\usepackage{url}
\usepackage{graphicx} 
\usepackage[colorlinks=true,linkcolor=blue,citecolor=blue]{hyperref}

\usepackage{algorithm}
\usepackage[noend]{algpseudocode}
\usepackage{amsmath,amssymb}
\usepackage{amsfonts}
\usepackage{amsthm}
\usepackage{bm}
\usepackage{multicol}
\theoremstyle{definition}
\newtheorem{definition}{Definition}[section]
\newcommand{\iid}{\overset{\mathrm{i.i.d.}}{\sim}}
\newcommand{\mockalph}[1]{}
\newtheorem{theorem}{Theorem}[section]
\newtheorem{proposition}[theorem]{Proposition}
\newtheorem{refremark}{Remark}[section]

\newcommand{\p}{\mathbb{P}}
\newcommand{\q}{\mathbb{Q}}

\IfFileExists{upquote.sty}{\usepackage{upquote}}{}
\IfFileExists{microtype.sty}{
  \RequirePackage[nopatch={eqnum,footnote}]{microtype}
  \UseMicrotypeSet[protrusion]{basicmath} 
}{}
\makeatletter
\@ifundefined{KOMAClassName}{
  \IfFileExists{parskip.sty}{%
    \usepackage{parskip}
  }{
    \setlength{\parindent}{0pt}
    \setlength{\parskip}{6pt plus 2pt minus 1pt}}
}{
  \KOMAoptions{parskip=half}}
\makeatother
\providecommand{\tightlist}{%
  \setlength{\itemsep}{0pt}\setlength{\parskip}{0pt}}\usepackage{longtable,booktabs,array}

\jname{Statistics and Its Application}
\jvol{AA}
\jyear{2024}
\doi{10.1146/((please add article doi))}

\begin{document}

\markboth{Laub et al.}{Hawkes Models And Their Applications}

\title{Hawkes Models And Their Applications}

\author{Patrick J. Laub,$^1$ Young Lee,$^2$ Philip K. Pollett,$^3$ and Thomas Taimre$^3$
\affil{$^1$School of Risk and Actuarial Studies, University of New South Wales, Sydney, Australia, 2052; email: p.laub@unsw.edu.au}
\affil{$^2$Department of Statistics and Data Science, National University of Singapore, Singapore, Singapore, 119244}
\affil{$^3$School of Mathematics and Physics, University of Queensland, Brisbane, Australia, 4072}}

\begin{abstract}
The Hawkes process is a model for counting the number of arrivals to a system which exhibits the \emph{self-exciting} property --- that one arrival creates a heightened chance of further arrivals in the near future.
The model, and its generalizations, have been applied in a plethora of disparate domains, though two particularly developed applications are in seismology and in finance.
As the original model is elegantly simple, generalizations have been proposed which: track marks for each arrival, are multivariate, have a spatial component, are driven by renewal processes, treat time as discrete, and so on.
This paper creates a cohesive review of the traditional Hawkes model and the modern generalizations, providing details on their construction, simulation algorithms, and giving key references to the appropriate literature for a detailed treatment.
\end{abstract}

\begin{keywords}
self-exciting point process, conditional intensity, Poisson cluster process, space-time ETAS, renewal Hawkes process, interval-censored data, discrete-time Hawkes process, dynamic contagion process 
\end{keywords}
\maketitle


\section{Introduction}\label{introduction}

The key property of Hawkes processes, called \emph{self-excitation}, is that the occurrence of one event may trigger a series of similar events. For example, earthquakes trigger aftershocks, social media posts generate reactionary posts, gang violence produces retaliatory attacks, a neuron fires, triggering brain activity elsewhere, infections propagate, or stock-market transactions may trigger a chain reaction. They were introduced by Alan Hawkes in 1971 \citep{hawkes1971spectra}, and now bear his name. Commenting on the initial impact of this work, he wrote ``those self-exciting processes did not seem to have generated all that much excitement'' \citep{hawkes2021personal}. The idea was taken up by the seismology community \citep{ogata1978} who adapted it to earthquake modelling (we cover this in Sections~\ref{sec-marked-hawkes} and \ref{sec-spatiotemporal-hawkes}). The explosion in popularity of Hawkes processes in the last 20 years may possibly be attributed to financial applications \citep{hawkes2018hawkes}.

The basic Hawkes process is a stochastic process \(N(t)\) that counts the number of events (termed ``arrivals'') up to time \(t\), increasing in steps of size~1 at random times.  Self-excitation is effected by allowing the rate of arrivals to depend in a particular way on past history, an idea that is made precise in Section~\ref{the-hawkes-process-and-the-self-exciting-property}.

In one of several extensions of the basic Hawkes process, a random variable called a \emph{mark} is associated with each arrival.  For example, earthquakes occur at a specific time, but also have an epicentre in space (latitude, longitude, and depth) and have an associated magnitude. Buy or sell orders in financial markets have a time and a volume associated with them (along with much other metadata). These \emph{marked Hawkes processes} are discussed in Section~\ref{sec-marked-hawkes}.

Or, we may have multiple Hawkes processes which are \emph{mutually exciting}.  For example, in a social network, the arrival of a message may trigger more messages from the same user, but also from other users. These \emph{mutually exciting Hawkes processes} are discussed in Section~\ref{sec-mutually-exciting-hawkes}.

We discuss some of these generalizations (and combinations of these generalizations!), but there is huge breadth, and we cannot cover all of them. Others who have reviewed Hawkes process scholarship have grappled with this.  Of note are works of \citet{hawkes2018hawkes} and \citet{worrall2022fifty}, which are particularly recommended for their treatments of financial applications and Bayesian non-parametric methods respectively.

The reader is not assumed to have a background in Hawkes processes, but some background in probability theory and stochastic processes is assumed.  Section~\ref{sec-classical-hawkes} provides background on the early forms of Hawkes processes, sufficient to understand the modern extensions in Section~\ref{sec-modern-hawkes}. Expert Hawkes process readers may wish to skim Section~\ref{sec-classical-hawkes} (to learn the notation used) and go straight to Section~\ref{sec-modern-hawkes}. The split of content between ``classical'' (Section~\ref{sec-classical-hawkes}) and ``modern'' (Section~\ref{sec-modern-hawkes}) Hawkes processes is mostly historical in nature, though it also reflects our subjective view of the most coherent order to present this material and (in some cases) of the relative conceptual difficulty.

\newpage
\section{Classical Hawkes Processes}\label{sec-classical-hawkes}

The first main section of the paper will briefly go over the key concepts of the simplest (univariate, linear, exponential decay) Hawkes processes.
It will stay at an advanced level and provide key references to the reader wanting to find details and proofs.
This section can be seen as a condensed summary of the key parts of our book \citep{laub2022elements}.

\subsection{Counting and point processes}\label{counting-and-point-processes}

We begin with the basic notation and concepts required in order to formally define the classical Hawkes process, starting with the \emph{counting process} and the \emph{point process}.
There are quite general definitions for these (e.g.~spatiotemporal), but to define the classical Hawkes we only need the most basic variant.

\begin{definition}[]\protect\hypertarget{def-point-process}{}\label{def-point-process}

A stochastic process \(N(t)\), defined for \(t \ge 0\), is a \emph{counting process} starting at \(N(0) = 0\) if \(N(t)\) only takes values in \(\{0, 1, 2, \dots\}\) and increases in jumps of size \(+1\).
The random jump times \(T_1\), \(T_2\), \(\dots\) form a \emph{point process}, with \(0 < T_1 < T_2 < \dots\) a.s.
The counting process can be defined in terms of the point process as
\[
    N(t) = \sum_{i=1}^{\infty} \mathbb{I}\{T_i \le t\} = \sum_{T_i \le t} 1 \,.
\]
The point processes considered are \emph{simple} in that there cannot be more that one jump at any given time.
We write \(N_t\) and \(N(t)\) interchangeably.
\hfill\(\diamond\)

\end{definition}

Commonly one finds the related notation and nomenclature:

\begin{itemize}
\tightlist
\item
  the \(T_i\) jump times are also called \emph{arrival times},
\item
  \(T_0 \equiv 0\) though this would not be counted by \(N_t\), and
\item
  the \(E_i := T_i - T_{i-1}\) are called the \emph{inter-arrival times}.
\end{itemize}

Perhaps the simplest counting process is the \emph{(homogeneous) Poisson process}, which has \(E_i \iid \mathsf{Exponential}(\lambda)\), \(\mathbb{E}(E_i)=1/\lambda\), and \(\mathbb{E}(N_t/t)=\lambda\).

\subsection{Conditional intensity and compensators}\label{conditional-intensity-and-compensators}

\begin{definition}[]\protect\hypertarget{def-conditional-intensity-function}{}\label{def-conditional-intensity-function}

The \emph{conditional intensity process} of \(N_t\) is defined, for \(t \ge 0\), by
\begin{equation}\phantomsection\label{eq-conditional-intensity-definition}{
    \lambda^*(t) = \lim_{\Delta \searrow 0} \frac{\mathbb{E}(N_{t+\Delta} - N_t \,\mid\, \mathcal{H}_t)}{\Delta} \,,
}\end{equation}
if this limit exists.
Here \(\mathcal{H}_t\), the \emph{history} of \(N_t\), is the filtration \citep[cf.][]{billingsley2017probability} generated by \(\{N_s\}_{s < t}\).
We use \(\lambda^*_t\) and \(\lambda^*(t)\) interchangeably.
\hfill\(\diamond\)

\end{definition}

It is important to realise that \(\lambda^*_t\) is generally a random process, representing the instantaneous rate of arrival at time \(t\), given the arrival times up to (but not including) time \(t\); the superscript~${}^*$ is used to remind us of this fact \citep[p.~249]{daley2003a}.
However, if \(\lambda^*(t) \equiv \lambda(t)\) for some function \(\lambda(t)\), so that the conditional intensity does not depend on the history, then \(N_t\) is the Poisson process with rate function \(\lambda(t)\); \emph{homogeneous} if \(\lambda(t) \equiv \lambda\) does not depend on \(t\).

Note that \(N_t\) is right-continuous with left-limits while \(\lambda^*_t\) is left-continuous with right-limits; this becomes important when it comes to correctly evaluating the likelihood of a point process.

When we integrate the conditional intensity up to time \(t\) we obtain another random process.

\begin{definition}[]\protect\hypertarget{def-compensator}{}\label{def-compensator}

The \emph{compensator} of the point process \(N_t\) is defined for \(t \ge 0\) by

\begin{flalign*}
&& \Lambda_t &= \int_0^t \lambda^*_s \, \mathrm{d}s \,. && \diamond
\end{flalign*}

\end{definition}

Indeed \(\Lambda_t\) \emph{characterises} the point process \citep[Proposition 14.1.VI of][]{daley2003b}; it is the unique predictable process for which \(N_t-\Lambda_t\) is a local martingale.
If \(\Lambda_t\) is absolutely continuous with respect to \(t\), then its path-wise derivative is \(\lambda^*_t\) and admits the interpretation of Equation~\ref{eq-conditional-intensity-definition}; see \citet{BrownPollett1982} and the references therein.
Furthermore, \(N_t\) is a Poisson process if and only if \(\Lambda_t\) is a deterministic (non-random) function of \(t\) \citep[see, for example,][]{BrownPollett1991}.

\subsection{The Hawkes process and the self-exciting property}\label{the-hawkes-process-and-the-self-exciting-property}

\citet{hawkes1971spectra} defined the counting process which is now named after him by specifying its conditional intensity process.

\begin{definition}[]\protect\hypertarget{def-hawkes-process}{}\label{def-hawkes-process}

A \emph{Hawkes process} is a counting process \(N_t\) whose conditional intensity process for \(t \ge 0\) is
\begin{equation}\phantomsection\label{eq-hawkes-intensity}{
    \lambda^*_t = \lambda + \sum_{T_i < t} \mu(t - T_i),
}\end{equation}
where \(\lambda > 0\) is the \emph{background arrival rate} and \(\mu: \mathbb{R}_+ \to \mathbb{R}_+\) is the \emph{excitation function}.
\hfill\(\diamond\)

\end{definition}

The excitation function controls how past arrivals will affect the rate of future arrivals.
One simple and quite special excitation function is \(\mu(t) = \alpha \exp( - \beta t )\), i.e.~exponentially decaying intensity.
The top row of Figure~\ref{fig-simulate-hawkes-exp-and-nonlinear} shows a simulated Hawkes process \(N_t\) and the observed exponentially decaying intensity \(\lambda^*_t\).
These subfigures highlight the fundamental Hawkes process idea of \emph{self-excitation}, i.e., that the arrival of an event can trigger more events in the short-term future.
This entails clusters of arrivals more frequently than would be found in a Poisson process.

\begin{figure}[h]
\centering{
\includegraphics[width=\textwidth]{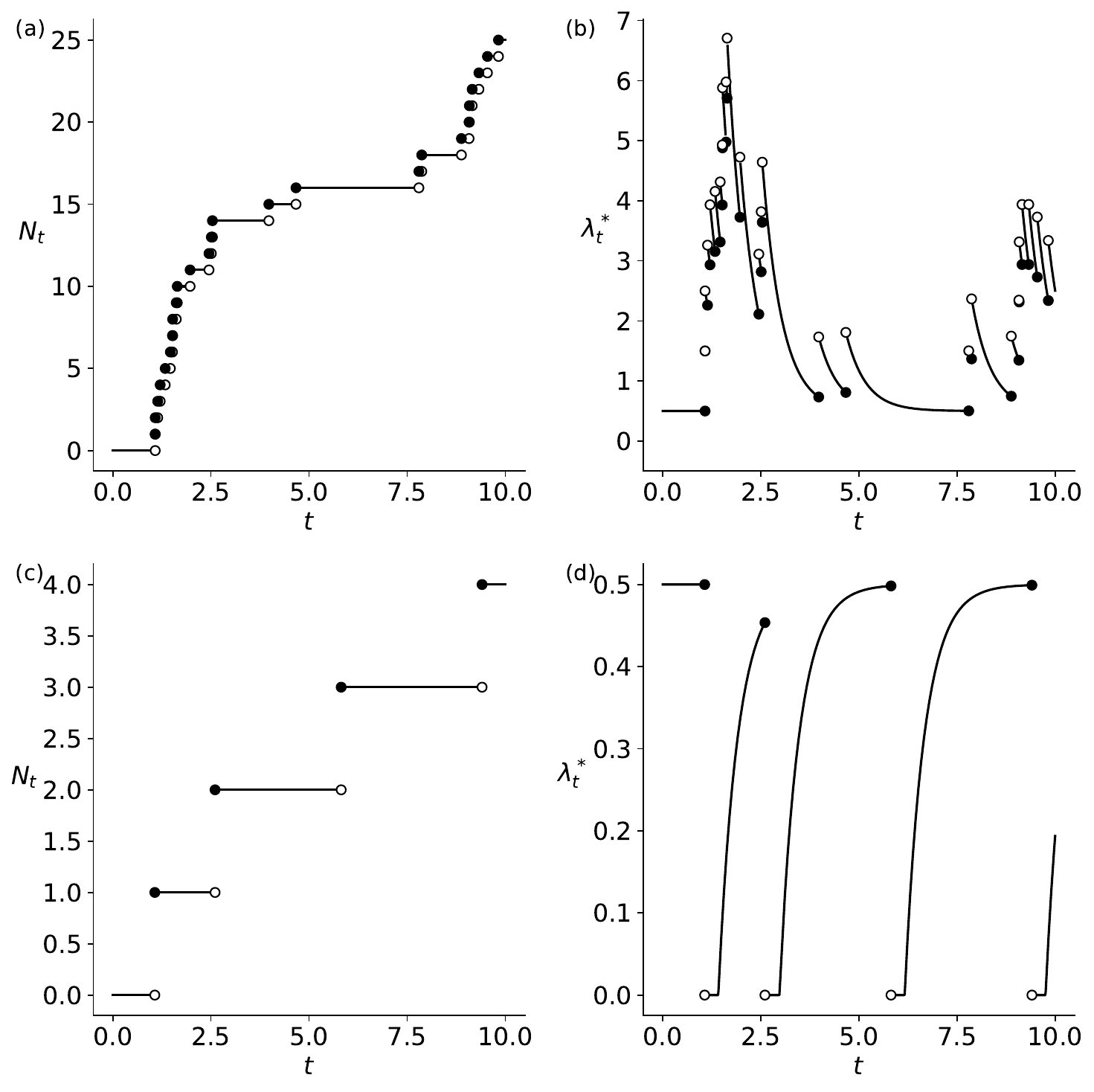}
}
\caption{\label{fig-simulate-hawkes-exp-and-nonlinear}Two Hawkes processes.
The top row shows a self-exciting Hawkes process and the bottom row a self-inhibitory nonlinear Hawkes process (see Section~\ref{sec-nonlinear-hawkes}).
Panels (\emph{a}) and (\emph{c}) show the counting processes \(N_t\) while panels (\emph{b}) and (\emph{d}) show the conditional intensity processes \(\lambda^*_t\).
Both have \(\lambda=\frac12\) and the top process has \(\mu(t) = \exp(-2 t)\) while the bottom process has \(\mu(t) = {-} \exp(-2 t)\) and \(\phi(x) = \max(\lambda + x, 0)\).
}
\end{figure}%

Note that basically every aspect of Definition~\ref{def-hawkes-process} can be or has been generalized in the literature.
The constant background rate \(\lambda\) can be replaced with a time-varying version to account for long-term trends or seasonality \citep[e.g.][]{lewis2012}.
The jumps sizes defined by \(\mu\) which are deterministic can instead be based on random marks (see Section~\ref{sec-marked-hawkes}) or more general stochastic processes \citep[e.g.][]{lee2016hawkes}.
The intensity can be a nonlinear transformation of the right-hand-side of Equation~\ref{eq-hawkes-intensity} (see Section~\ref{sec-nonlinear-hawkes}).
Instead of time \(t \in [0, \infty)\) we can have \(t \in \mathbb{R}\) or even discrete time \(t \in \mathbb{N}_0\) \citep[e.g.][]{white2013a}, see Section~\ref{sec-discrete-time-hawkes}.

\clearpage

\subsection{The immigration-birth view and stationarity}\label{sec-immigration-birth-view}

While the Hawkes process can be understood in terms of its conditional intensity in Equation~\ref{eq-hawkes-intensity}, there's also an intuitive construction of the process as a Poisson cluster process called the \emph{immigration-birth} representation \citep{hawkes1974}.

\begin{multicols}{2}

In this view, we consider a homogeneous Poisson process with rate \(\lambda\) as the \emph{immigration process}.
Immigrants represents exogenous shocks to the system.
An immigrant arriving at time \(s\) then creates a new Poisson process of \emph{births} or \emph{offspring} with intensity \(\mu(t - s)\) for \(t > s\).
Each of these births may then generate more offspring in the same manner recursively, hence we generate the endogenous self-exciting behaviour.
The Hawkes process is then the aggregation of the immigration and these birth processes.
Algorithm~\ref{alg:immig-birth-simulation-compact} shows how to simulate a Hawkes process based on this construction.\footnote{Algorithm~\ref{alg:immig-birth-simulation-compact} should be equivalent to Algorithm~3 in \citet{laub2015hawkes}, however that earlier version has an error as it only generates offspring of immigrants.
  Unfortunately this error was propagated to Algorithm~3 of \citet{laub2022elements} and Algorithm~2 of \citet{worrall2022fifty}.}
It relies on a Poisson process simulator called \(\texttt{SimPoisson}\) which takes an intensity function and a time horizon.

\columnbreak

\begin{algorithm}[H]
\caption{Simulate by Immigrant-Birth}
\label{alg:immig-birth-simulation-compact}
\begin{algorithmic}
\Require $\theta = (\lambda, \mu)$, $T$
\State \texttt{intens}$(s) \gets (s \mapsto \lambda)$
\State \texttt{imm} $\gets$ \Call{SimPoisson}{\texttt{intens}, $(0, T)$}
\Function{Births}{$t$, $T$, $\mu$}
\State \texttt{intens}$(s) \gets (s \mapsto \mu(s-t))$
\State \texttt{offsp} $\gets$ \Call{SimPoisson}{\texttt{intens}, $(t, T)$}
\For{$t_j$ in \texttt{offsp}}
\State \texttt{offsp} $\gets$ \texttt{offsp} $\cup$ \Call{Births}{$t_j, T, \mu$}
\EndFor
\State \Return \texttt{offsp}
\EndFunction

\State \texttt{arr} $\gets \emptyset$
\For{$t_i$ in \texttt{imm}}
\State \texttt{arr} $\gets$ \texttt{arr} $\cup \, \{ t_i \} \, \cup$ \Call{Births}{$t_i$, $T$, $\mu$}
\EndFor

\State \Return Sort(\texttt{arr})
\end{algorithmic}
\end{algorithm}

\end{multicols}

The immigration-birth representation means that we can analyse the Hawkes process using the machinery of branching processes.
An arrival at time \(s\) will generate \(\mathsf{Poisson}(\eta)\) number of first-generation offspring over \(t \ge s\).
This \(\eta\), called the \emph{branching ratio}, is defined by
\[
    \eta := \int_{s}^\infty \mu(t - s) \, \mathrm{d}t  = \int_0^\infty \mu(t) \, \mathrm{d}t \,.
\]
Each of these (on average) \(\eta\) offspring will then generate their own offspring.
The expected number of offspring in the \(k\)-th generation is then \(\eta^k\), hence the expected number of offspring in all generations is \(\sum_{k=1}^\infty \eta^k = \eta/(1 - \eta)\) if \(\eta < 1\).
The condition \(\eta < 1\) is necessary and sufficient for the process to be stationary and have finite mean.
This ameliorates the concern that the self-exciting nature of the Hawkes process may lead to an explosion of arrivals (i.e.~that \(\mathbb{E}[N_t] = \infty\) for some finite time \(t\)).
For example, if \(\mu(t) = \alpha \exp(-\beta t)\) then the expected number of direct offspring \(\eta = \alpha / \beta\) and hence \(\alpha < \beta\) is required to ensure stationarity.

The branching ratio \(\eta\) permeates the Hawkes process literature.
It is in many limit theorems such as the law of large numbers \citep{daley2003a}
\[
    \lim_{t \to \infty} \frac{N_t}{t} = \frac{\lambda}{1 - \eta} \quad \text{a.s.}
\]
As \(\mu(t) = \eta \, \nu(t)\) where \(\nu\) is a probability density function (p.d.f.), then it is (arguably) more interpretable to model \(\eta\) separately from \(\nu\) \citep[as in][]{bacry2020sparse,porter2012self}.

\subsection{Excitation functions and Markov analysis}\label{excitation-functions-and-markov-analysis}

The choice of which parametric form for the excitation function \(\mu\) is a key decision when modelling using a Hawkes process.
The function can be chosen to incorporate some desired characteristics such as a delayed impact of an arrival on the intensity; this can be done for \(\mu(t) = \eta \nu(t)\) with \(\nu(t)\) set as the p.d.f. of a gamma distribution \citetext{\citealp[cf.][]{cui2022moments}; \citealp[insurance applications in][]{lesage2022hawkes}}.
Or \(\mu\) can be chosen to encode domain-specific expertise, as in the following example.

\begin{definition}[]\protect\hypertarget{def-omori-utsu-law}{}\label{def-omori-utsu-law}

The 130~year-old \emph{Omori law} in seismology models an earthquake's aftershock rate as approximately \(K (t + c)^{-1}\) (\(K, c > 0\)) \citep{omori1894, utsu1995centenary}.
\citet{utsu1961statistical} created the \emph{modified Omori law} \(\mu(t) = K (t + c)^{-p}\) (\(p > 1\)) --- also called the \emph{Omori--Utsu law} --- which is still used in modern earthquake models (e.g.~in Definition~\ref{def-etas}'s ETAS model).
The integral of the Omori--Utsu law gives \(\eta = K c^{1-p} / (p-1)\).

This \(\mu\) is also called the \emph{power-law kernel} since its normalized form \(\nu(t) = \frac{p-1}{c} (1 + \frac{t}{c})^{-p}\) is a variant of the \(\mathsf{Pareto}(p-1)\) distribution.
This famously heavy-tailed distribution may seem inappropriate to those outside of seismology, though \citet{utsu1995centenary} states that ``aftershock activities following large earthquakes continue ten years or more.''
\hfill\(\diamond\)

\end{definition}

However the case of the exponentially decaying \(\mu(t) = \alpha \exp(-\beta t)\) is by far the most common choice in the literature.
The exponential decay is special because \((N_t, \lambda^*_t)\) is a \emph{Markov process} in this case \citep[Remark 1.22 of][]{liniger2009}.
In particular, between jumps \(T_n \le t < T_{n+1}\) the exponentially decaying intensity is
\begin{equation}\phantomsection\label{eq-exp-markov-decay}{
    \lambda^*_t = \lambda + (\lambda^*_{T_n+} - \lambda) \, \mathrm{e}^{-\beta (t - T_n)} = \lambda + (\lambda^*_{T_n} + \alpha - \lambda) \, \mathrm{e}^{-\beta (t - T_n)} \,.
}\end{equation}
This joint Markovian property greatly simplifies inference, simulation, and analysis of the Hawkes process, and explains why the exponential decay is, in some sense, the default choice for \(\mu\).
Other recursive forms for \(\lambda^*_t\) are available in certain cases, e.g.~for \(\mu(t) = \sum_{i=1}^p \alpha_i t^i \exp(-\beta t)\) \citep{ozaki1979generalized, ogata1981}.

Given an exponentially decaying \(\mu\), the first natural extension is to allow the conditional intensity to start at a value \(\lambda_0\) which is different from the background rate \(\lambda\).

\begin{definition}[]\protect\hypertarget{def-exp-hawkes-process}{}\label{def-exp-hawkes-process}

An \emph{exponentially decaying Hawkes process} is a counting process \(N_t\) defined by \(\lambda_0, \lambda, \alpha, \beta > 0\) whose conditional intensity starts at \(\lambda_0\) and for \(t \ge 0\) follows
\[
    \lambda^*_t = \lambda + (\lambda_0 - \lambda) \, \mathrm{e}^{-\beta t} + \sum_{T_i < t} \alpha \, \mathrm{e}^{-\beta (t - T_i)} \,, \quad \text{ hence }
\]
\begin{equation}\phantomsection\label{eq-exp-markov-compensator}{
\begin{aligned}
    \Lambda_t 
    &= \lambda t + \frac{(\lambda_0 - \lambda)}{\beta} (1 - \mathrm{e}^{-\beta t}) + \sum_{T_i < t} \frac{\alpha}{\beta} (1 - \mathrm{e}^{-\beta (t - T_i)}) \,.
\end{aligned}
}\end{equation}
\hfill\(\diamond\)

\end{definition}

While this may seem like a minor adjustment (adding one more parameter), it represents a notable conceptual shift.
The earlier counting processes describe systems which have a clear starting time \(t=0\) and which have been observed continuously since time 0.
However a multitude of data-generating processes either have no clear start time or have not been observed since \(t=0\) (our meteorological records don't include the first hurricane, the first earthquake, etc.).
Definition~\ref{def-exp-hawkes-process} builds on the Markov property to analyse these long-running systems where \(t=0\) now represents the time when records began.

\subsection{Likelihood function}\label{likelihood-based-inference}

The likelihood for any point process parameterized by \(\theta\) with observations \(\{t_1, \dots, t_n\}\) within a time horizon \(0 \le t \le T\) is
\[
    L(\theta \mid t_1, \dots, t_n, T) = \Biggl[ \prod_{i=1}^n \lambda^*_{t_i} \Biggr] \exp\Bigl(-\int_0^T \lambda^*_t \, \mathrm{d}t\Bigr) = \Biggl[ \prod_{i=1}^n \lambda^*_{t_i} \Biggr] \mathrm{e}^{-\Lambda_T} \,, \text{ hence }
\]
\begin{equation}\phantomsection\label{eq-log-likelihood}{
    \ell(\theta \mid t_1, \dots, t_n, T) := \log( L(\theta \mid t_1, \dots, t_n, T ) ) = \sum_{i=1}^n \log \lambda^*_{t_i} - \Lambda_T \,.
}\end{equation}
Thus the log-likelihood for a Hawkes process with a general excitation function is
\begin{equation} \label{eq-hawkes-likelihood}
    \ell = \sum_{i=1}^n \log \Bigl( \lambda + \sum_{t_j < t_i} \mu(t_i - t_j) \Bigr) - \Lambda_T \,.
\end{equation}

\begin{multicols}{2}
Equation~\ref{eq-hawkes-likelihood} for the Hawkes process log-likelihood can be evaluated quite easily, though not necessarily quickly.
It contains a double summation over all pairs of event times \(t_i\) and \(t_j\), leading to \(\mathcal{O}(N_T^2)\) complexity.
For exponentially decaying Hawkes however, Equation~\ref{eq-exp-markov-decay} allows us to compute the intensity at each event time in linear time.
This fast \(\mathcal{O}(N_T)\) method, derived by \citet{ozaki1979}, is shown in Algorithm~\ref{alg:exp-hawkes-likelihood}.

\columnbreak

\begin{algorithm}[H]
\caption{Log-Likelihood}
\label{alg:exp-hawkes-likelihood}
\begin{algorithmic}
\Require $(\lambda_0, \lambda, \alpha, \beta)$, $\{t_1, t_2, \ldots, t_n\}$, $T$
\State $\ell \gets 0$, $t_0 \gets 0$, $\lambda^*_0 \gets \lambda_0$
\State $\Lambda_T \gets \lambda T + \frac{(\lambda_0 - \lambda)}{\beta} (1 - \mathrm{e}^{-\beta T})$
\For{$i \gets 1$ to $n$}
    \State $\lambda^*_{t_i} \gets \lambda + (\lambda^*_{t_{i-1}} + \alpha - \lambda) \, \mathrm{e}^{-\beta (t_i - t_{i-1})}$
    \State $\Lambda_T \gets \Lambda_T + (\alpha/\beta) \, (1 - \mathrm{e}^{-\beta (t_i - t_{i-1})})$
    \State $\ell \gets \ell + \log(\lambda^*_{t_i})$
\EndFor
\State $\ell \gets \ell - \Lambda_T$
\State \Return $\ell$
\end{algorithmic}
\end{algorithm}

\end{multicols}

\subsection{Frequentist inference techniques}\label{sec-inference}

Traditional inference techniques for the Hawkes process include maximum likelihood estimators (MLE), expectation-maximization (EM) algorithms, and the generalized method of moments (GMM).
For the basic Hawkes model, MLEs have become the ``standard method'' \citep{Bacry2012Nonparametric} for estimation, as they are unbiased and are asymptotically efficient.

There are no explicit solutions for the maximizer of the Hawkes log-likelihood in Equation~\ref{eq-hawkes-likelihood}, but we can numerically maximize the likelihood in a process sometimes referred to as \emph{Direct Numerical Maximization (DNM)}.
This can be done with a generic optimization algorithm such as the Nelder--Mead algorithm \citep{nelder1965simplex}, or similar second-order method using the Hessian given by \cite{ozaki1979}.
However, \cite{Veen-Schoenberg:2008} note that ``in cases where the log-likelihood function is extremely flat in the vicinity of its maximum, such optimization routines can have convergence problems and can be substantially influenced by arbitrary choices of starting values.''

The EM algorithm is an iterative technique to numerically maximize the likelihood, though with special properties such as that the likelihood will monotonically increase over the iterations \citep{Dempster77maximumlikelihood}.
It estimates the parameters of the process while trying to infer some latent structure of the data, and for Hawkes processes
the ``EM algorithm for maximum likelihood \dots\ is often used on the branching process representation of the Hawkes process'' \citep{hawkes2021personal}.
\cite{Veen-Schoenberg:2008} give a specific EM method for the spatiotemporal Hawkes model, which is shown in simplified form in \citet[Chapter 6]{laub2022elements} for the basic model (even in simplified form, the notation required to describe it is too much for this paper). 
\cite{lapham2014hawkes} compares DNM to the EM method, and gives a thorough treatment of the merits of each method (recommending DNM for their specific application).
For some variations of the Hawkes model (based on the Poisson cluster interpretation) we are forced to use an EM algorithm; specifically, the renewal Hawkes process requires an EM approach for inference.

This discussion has assumed that we can observe the exact arrival times for the process, but in many applications we only have access to binned data.
Section~\ref{sec-binned-data} is devoted to discussing solutions to this problem.
Specifically, the GMM approach is used in this scenario.

\subsection{Simulation algorithms}\label{simulation-algorithms}

The ability to simulate Hawkes processes allows for the calculation (or approximation) of a great deal of useful quantities, e.g.~prediction intervals, parametric bootstrap estimates, sensitivity analyses.
We have already seen one simulation method, in Algorithm~\ref{alg:immig-birth-simulation-compact}, though that was shown for pedagogical reasons and is not an efficient algorithm.

One widely applicable technique is \emph{Ogata's modified thinning method} \citep{ogata1981} shown in Algorithm~1 in the supplementary materials.
It is a variation of the thinning algorithm for inhomogeneous Poisson processes \citetext{\citealp{lewis1979}; \citealp[cf.~Algorithm 5.12 of][]{kroese2011}}.
This method relies on us supplying upper bounds for \(\lambda^*_t\) between jumps.
These bounds are normally easy to find for linear Hawkes, e.g.~for exponential Hawkes
\(\lambda^*_t \le \lambda^*_{T_n} + \alpha\) for \(T_n \le t < T_{n+1}\).
Even for the nonlinear Hawkes in the bottom row of Figure~\ref{fig-simulate-hawkes-exp-and-nonlinear} we can use thinning since \(\lambda^*_t \le \lambda\) for all~\(t\).
When the intensity is increasing without asymptote however, then a small modification to the thinning method is required \citep[e.g.~as in][]{lee2022exact}.
Similarly, the simulation of mutually-exciting Hawkes processes can be achieved with only a small modification to the thinning method.

As is usual with Hawkes processes, the exponentially-decaying Hawkes process has a much more efficient algorithm than the general case.
The algorithm for exponentially-decaying Hawkes is called the \emph{exact simulation method} \citep{dassios2013}, or it can be seen as applying the \emph{composition method} to the interarrival distribution which is Gompertz--Makeham distributed \citep{pai1997}. See Algorithm~2 in the supplementary materials for the pseudocode.

\subsection{Nonlinear Hawkes}\label{sec-nonlinear-hawkes}

\citet{bremaud1996} generalized the relationship between the sum of excitation functions and the conditional intensity function by defining the following.

\begin{definition}[]\protect\hypertarget{def-nonlinear-hawkes-process}{}\label{def-nonlinear-hawkes-process}

A \emph{nonlinear Hawkes process} is a counting process \(N_t\) whose conditional intensity process for \(t \ge 0\) is
\begin{equation}\phantomsection\label{eq-nonlinear-hawkes-intensity}{
    \lambda^*_t = \phi\Bigl( \sum_{T_i < t} \mu(t - T_i) \Bigr),
}\end{equation}
where \(\phi : \mathbb{R} \to \mathbb{R}_+\) is a (possibly) nonlinear function and \(\mu: \mathbb{R}_+ \to \mathbb{R}\) is the excitation function.
The classical `linear' Hawkes is recovered if \(\phi(x) = \lambda + x\).
\hfill\(\diamond\)

\end{definition}

Adding this nonlinear function increases the flexibility of the Hawkes model quite substantially, to the point that nonlinear Hawkes may not even exhibit the fundamental Hawkes property of self-excitation.
For example, we can choose a \(\phi\) which induces the opposite of self-excitation, a \emph{self-inhibitory} effect where the arrival of an event can suppress future events and make arrivals appear in a more regular clock-like pattern.
One specific form is to let \(\mu(t) = \alpha \exp(-\beta t)\) with \(\alpha < 0\) (and \(\beta > 0\)) which creates negative jumps in the intensity and then choose \(\phi(x) = \max(\lambda + x, 0)\) to ensure non-negativity of \(\lambda^*_t\) \citep[cf.][]{bonnet2021maximum}.
This nonlinear Hawkes process is shown in the bottom row of Figure~\ref{fig-simulate-hawkes-exp-and-nonlinear}, and when contrasted with the self-excited Hawkes in the top row we can clearly see the self-inhibitory effect.

\subsection{Mutually exciting Hawkes}\label{sec-mutually-exciting-hawkes}

The Hawkes process can be extended to the multivariate setting, where we maintain counting processes for multiple different kinds of arrivals.

\begin{definition}[]\protect\hypertarget{def-mutually-exciting-hawkes}{}\label{def-mutually-exciting-hawkes}

Consider \(\bm{N}_t = (N^1_t, \dots, N^d_t)\) as a collection of \(d\) counting processes with \(N^k_t\)'s arrival times denoted by \(T^k_1, T^k_2\), etc.
They are \emph{mutually-exciting Hawkes processes} if \(N^k_t\)'s conditional intensity follows
\[
    \lambda^*_k(t) = \lambda + \sum_{j=1}^d \sum_{T^j_i < t} \mu_{j,k}(t - T^j_i) \quad \text{ for } k=1,\dots,d \,,
\]
where \(\mu_{j,k}(s) \ge 0\).
\hfill\(\diamond\)

\end{definition}

\begin{refremark}
Denote the integrals of the \(\mu_{j,k}\) excitation functions as \(\phi_{j,k} = \int_0^\infty \mu_{j,k}(s) \,\mathrm{d}s\) and the matrix of them as \(\Phi = (\phi_{j,k}) \in \mathbb{R}^{d \times d}\).
The mutually-exciting Hawkes process \(\bm{N}_t\) is stationary (non-explosive) iff the spectral radius of \(\Phi\) is less than one.

\label{rem-mutually-exciting-hawkes-stationarity}

\end{refremark}

The \(N^j_t\) process has either an excitatory effect on \(N^k_t\) (\(\phi_{j,k} > 0\)) or no effect on \(N^k_t\) (\(\phi_{j,k} = 0\)).
The excitatory case \(\phi_{j,k} > 0\) includes self-excitation (\(j=k\)) and mutual excitation (\(j\not=k\)).
Thus the \(\Phi\) matrix is akin to the adjacency matrix of a directed graph with \(d\) nodes, or equivalently it encodes the causal relationships in this network of \(d\) counting processes.
These processes have \(\mathcal{O}(d^2)\) parameters to fit which is often undesirable for many applications.
One theme of Hawkes process scholarship, exemplified by \citet{bacry2020sparse} and \citet{cai2022latent}, has focused on inferring sparse \(\Phi\) matrices.

\section{Modern Hawkes Processes}\label{sec-modern-hawkes}

This largest section examines the newer variants of Hawkes which were not originally covered in our book \citep{laub2022elements}.
We will define each process and describe the key innovations that they bring, and (space-permitting) say some words (and mainly supply references) for each processes' inference, simulation, theoretical results, and so on.

\subsection{Marked Hawkes processes}\label{sec-marked-hawkes}

Marked Hawkes processes are perhaps the simplest extension of the classical Hawkes process framework.
A random variable called a `mark' is associated with each arrival in the underlying process.
Canonical examples include earthquakes and their magnitudes, financial transactions and their volumes, and insurance claims and their sizes.

Let \(M_i\in \mathcal{M}\) be the mark associated with the arrival at time \(T_i\).
Let us say that \(N([s, t), \mathcal{A})\) counts the number of arrivals with marks in \(\mathcal{A} \subset \mathcal{M}\) from time \(s\) up to \(t\).
Then our Definition~\ref{def-conditional-intensity-function} for the conditional intensity process becomes
\[
    \lambda^*(t, m) = \lim_{\Delta t, \Delta m \searrow 0} \frac{\mathbb{E}\bigl[ N\bigl( [t, t+\Delta t), [m + \Delta m) \bigr) \,\mid\, \boldsymbol{\mathcal{H}}_t \bigr]}{ \bigl| \Delta m \bigr| \bigl| \Delta t \bigr|} \,,
\]
if this limit exists (now \(\boldsymbol{\mathcal{H}}_t\) is the filtration generated by \(\{ (T_i, M_i) : T_i < t \}\)).
Typically this is written as
\[
    \lambda^*(t, m) = \lambda^*_\mathrm{g}(t) \, f^*(m \,|\, t) \,,
\]
where \(\lambda^*_\mathrm{g}(t)\) is the intensity of the underlying `ground' point process (conditioned on \(\boldsymbol{\mathcal{H}}_t\), so it may depend on the marks) and \(f^*(m \,|\, t)\) is the density of the mark \(m\) at time \(t\) (also conditioned on \(\boldsymbol{\mathcal{H}}_t\)), cf.~\citet[Eq.~7.3.3]{daley2003a}.
The likelihood function \citep{rasmussen2013, JosephJain2024} is then
\[
    L(\theta \mid (t_1, m_1), \dots, (t_n, m_n), T) = \Biggl[ \prod_{i=1}^n \lambda^*(t_i, m_i) \Biggr] \exp\Bigl(- \int_0^t \int_{\mathcal M} \lambda^*(s,m) \,\mathrm{d}m \,\mathrm{d}s \Bigr) \,.
\]
Numerical methods for maximising this likelihood have so far been found to be unstable, leading to investigation of expectation--maximisation procedures \citep{Veen-Schoenberg:2008} or Bayesian methods \citep{rasmussen2013}.

The simplest case is when the marks are modelled as completely independent of the arrival process.
This is traditionally the case in actuarial claims modelling, where the claim severity (marks) are i.i.d.\ and have no impact on the claim frequency (arrival process), cf.~the Cramér--Lundberg model \citep{cramer1930mathematical, lundberg1903approximerad}.
In this case the simulation and inference is straightforward; the Hawkes arrivals can be simulated/fitted as usual and then the marks can be simulated/fitted independently \citep[see e.g.][]{MollerRasmussen2006, JosephJain2024}.

The marks are called \emph{unpredictable} \citep[Definition 6.4.III]{daley2003a} if the marks are independent of the history ($f^*(m \,|\, t) = f(m \,|\, t)$) though they impact the underlying Hawkes intensity process.
Definition~\ref{def-etas}'s ETAS model is an example of this style of marked Hawkes process, and the next two definitions give the necessary context for ETAS's definition.

\begin{definition}[]\protect\hypertarget{def-gutenberg-richter-law}{}\label{def-gutenberg-richter-law}

Large magnitude earthquakes occur less frequently than smaller ones.
The \emph{Gutenberg--Richter} law in seismology observes that the magnitude of earthquakes can be modelled well by a shifted exponential distribution with p.d.f.
\[
    f(m) = \beta \mathrm{e}^{-\beta (m - m_0)} \quad \text{for } m \ge m_0,
\]
where \(\beta > 0\) and \(m_0 > 0\) is the minimum magnitude which can be detected.
\hfill\(\diamond\)

\end{definition}

\begin{definition}[]\protect\hypertarget{def-utsu-law}{}\label{def-utsu-law}

Large magnitude earthquakes tend to create more aftershocks than smaller ones.
The \emph{Utsu aftershock productivity law} \citep{utsu1969aftershocks} models the expected number of aftershocks (of at least magnitude \(m_0 > 0\)) following an earthquake of magnitude \(m\) to be
\[
    \eta(m) = A \mathrm{e}^{ \alpha (m - m_0) } \quad \text{for } m \ge m_0
\]
for some \(A, \alpha > 0\).
\hfill\(\diamond\)

\end{definition}

\begin{definition}[]\protect\hypertarget{def-etas}{}\label{def-etas}

The class of \emph{(temporal) epidemic-type aftershock sequence (ETAS)} models \citep{ogata1988} are marked Hawkes processes where the mark \(M_i\) represents the magnitude of the \(i\)-th earthquake.
Each earthquake generates aftershocks according to Definition~\ref{def-utsu-law}, so the ground intensity process is
\[
    \lambda^*_\mathrm{g}(t) = \lambda + \sum_{T_i < t} \eta(M_i) \, \nu(t - T_i) \,,
\]
where \(\nu\) is a p.d.f.
Specifically, the ETAS model uses
\[
\begin{aligned}
    \lambda^*(t, m) 
    &= f(m) \bigl[ \lambda + \sum_{T_i < t} \eta(M_i) \nu(t - T_i) \bigr] \\
    &= \beta \mathrm{e}^{-\beta (m - m_0)} \Bigl[ \lambda + \sum_{T_i < t} A \mathrm{e}^{ \alpha (M_i - m_0)} \frac{p-1}{c} \Bigl(1 + \frac{t - T_i}{c} \Bigr)^{-p} \Bigr] \,,
\end{aligned}
\]
for a given \(\lambda, A, \alpha, c > 0, p > 1\).
This means that the ETAS model satisfies all three empirical laws mentioned earlier: the Omori--Utsu law (Definition~\ref{def-omori-utsu-law}), the Gutenberg--Richter law (Definition~\ref{def-gutenberg-richter-law}), and the Utsu aftershock productivity law (Definition~\ref{def-utsu-law}).
It also uses \(f^*(m \,|\, t) = f(m)\), meaning that the magnitudes are modelled as independent of the history and of time.
\hfill\(\diamond\)

\end{definition}

The marks can be univariate or multivariate.
When multivariate, the overall process is what \citet{liniger2009} calls a `pseudo-multivariate' process since the arrival of multivariate marks may look like a multivariate process, but the underlying arrival process is still univariate (in contrast to the truly multivariate `mutually exciting' case in Definition~\ref{def-mutually-exciting-hawkes}).

Several examples mentioned later employ marks, for example the dynamic contagion process (see Section~\ref{sec-dynamic-contagion-process}) and spatiotemporal Hawkes processes (see Section~\ref{sec-spatiotemporal-hawkes}).

\subsection{Spatiotemporal Hawkes processes}\label{sec-spatiotemporal-hawkes}

In many systems of interest, the arrivals being counted occur at a particular time but also in space at some specific location.
Let's say that \(\bm{S}_i \in \mathbb{R}^d\) represents the spatial location of the arrival at \(T_i\).
Then the \emph{spatiotemporal counting process} \(N([s, t), \mathcal{R})\) is the number of arrivals in the spatial region \(\mathcal{R} \subset \mathbb{R}^d\) from time \(s\) up to \(t\).
If we let \(B(\bm{s}, r) \subset \mathbb{R}^d\) be a ball of radius \(r\) centred at \(\bm{s}\), then our Definition~\ref{def-conditional-intensity-function} for the conditional intensity process extends to the spatiotemporal setting \citep{reinhart2018review} as
\[
    \lambda^*(t, \bm{s}) = \lim_{\Delta s, \Delta t \searrow 0} \frac{\mathbb{E}\bigl[ N\bigl( [t, t+\Delta t), B(\bm{s}, \Delta s) \bigr) \,\mid\, \boldsymbol{\mathcal{H}}_t \bigr]}{ \bigl| B(\bm{s}, \Delta s) \bigr| \Delta t} \,,
\]
if this limit exists (now \(\boldsymbol{\mathcal{H}}_t\) is the filtration generated by \(\{ (T_i, \bm{S}_i, M_i) : T_i < t \}\)).
This is effectively a special case of the marked Hawkes processes discussed in Section~\ref{sec-marked-hawkes}, though there are some subtleties introduced.

We continue the seismology theme from the previous section.
If a region's seismic activity is observed to be lower than predicted by an earthquake model, a state called \emph{relative quiescence}, then this may be a sign that a large earthquake in the region is forthcoming \citep{ogata1998space}.
The temporal ETAS model from Definition~\ref{def-etas} has been used to this effect, however knowing simply that a large earthquake is imminent is not as useful as being able to predict the likely locations of the earthquake, which is the motivation for the space-time ETAS model below.

\begin{definition}[]\protect\hypertarget{def-spacetime-etas}{}\label{def-spacetime-etas}

The \emph{(space-time) epidemic-type aftershock sequence (ETAS)} model \citep{ogata1998space} are spatiotemporal marked Hawkes processes where the earthquake at time \(T_i\) has location \((X_i, Y_i)\) and magnitude \(M_i\).
The space-time ETAS intensity is
\[
\begin{aligned}
    \lambda^*(t, x, y, m) 
    &= f(m) \lambda^*(t, x, y) \\
    &= f(m) \bigl[ \lambda(x, y) + \sum_{i: T_i < t} \eta(M_i) \nu(t - T_i) g(x - X_i, y - Y_i; M_i) \bigr],
\end{aligned}
\]
where \(\lambda(x, y)\) is a spatially-aware background arrival rate and
\begin{equation}\phantomsection\label{eq-spatial-impact}{ 
    g(x, y; m) = \frac{q - 1}{\pi D \exp( \gamma (m - m_0) )} \bigl( 1 + \frac{x^2 + y^2}{D \exp( \gamma (m - m_0) ) } \bigr)^{-q}
}\end{equation}
is a p.d.f. which localises the impact on the intensity function to be near the arrival location \(X_i, Y_i\).
The functions \(\nu\), \(f\) and \(\eta\) are the same as in Definition~\ref{def-etas} (i.e.~given by Definitions~\ref{def-omori-utsu-law}, \ref{def-gutenberg-richter-law}, and \ref{def-utsu-law}).

\hfill\(\diamond\)

\end{definition}

The form of Equation~\ref{eq-spatial-impact} is flexible; here we show the one used by the R \texttt{ETAS} package \citep{jalilian2019etas}, though \citet{ogata1998space} lists other options.
Similarly the background intensity \(\lambda(x, y)\) can be fitted in many ways, though nonparametric approaches are common.

Since we simply observe earthquakes without knowing if they are immigrant events (called \emph{background} events) or offspring events (called \emph{triggered} events), then fitting the background intensity \(\lambda(x, y)\) while simultaneously finding all the parameters governing self-excitation is a serious challenge.
\citet{zhuang2002stochastic} propose an iterative method where --- given a current guess for \(\lambda(x, y)\) --- for each observation \(i\) we calculate
\[
    \rho_i = \mathbb{P}(\text{Event } i \text{ is a background event}) = \frac{\lambda(x_i, y_i)}{\lambda^*(t_i, x_i, y_i)}
\]
and then randomly assign observations as being background events with probability \(\rho_i\) or triggered events otherwise, and the fitting for \(\lambda(x, y)\) can be updated given this random set of `background events'.
This method is called \emph{stochastic declustering}.

As a demonstration, the Japanese earthquake dataset from \citet{laub2022elements} was fitted to the space-time ETAS model (previously a simple Hawkes process was employed).
Moving from a temporal model to a marked spatiotemporal model involved a significant jump in computing time: from 50 ms up to 2780 minutes (about 3 million \(\times\)).
Figure~\ref{fig-etas} shows the result of the \texttt{ETAS} package's \texttt{rates} function \citep{jalilian2019etas}, as well as a demonstration of \citet{zhuang2002stochastic}'s stochastic declustering method.
Note, it is common to augment \(\boldsymbol{\mathcal{H}}_t\) to include some extra events just outside of the spatiotemporal region of interest \citep[p.~370]{zhuang2002stochastic} to mitigate estimation errors caused by edge effects.

\begin{figure}[H]
    (\textit{a}) \\
    \includegraphics[width=\textwidth]{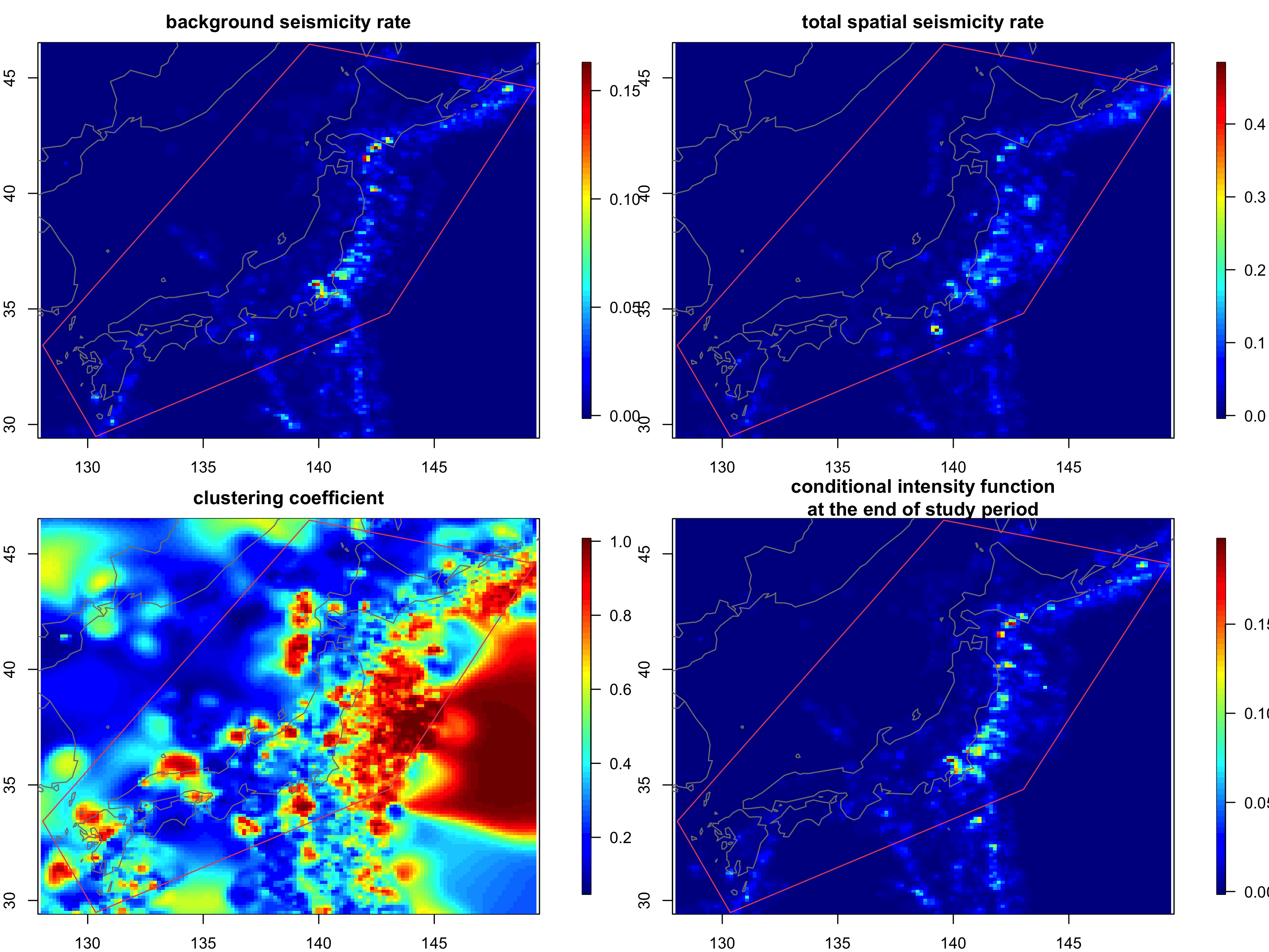}
    (\textit{b}) \\
    \includegraphics[width=\textwidth]{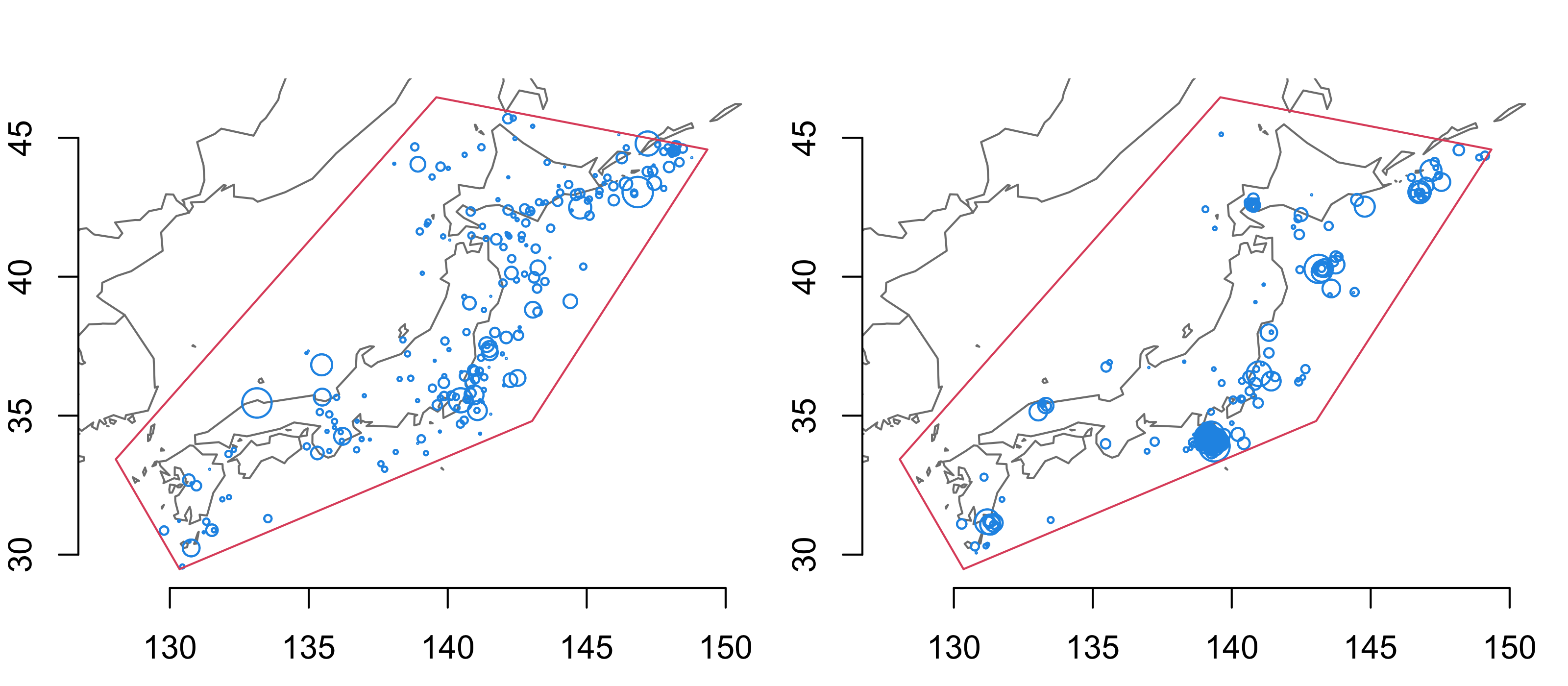}
    \caption{
        Results of the ETAS model analysis on earthquake data.
        Panel (\textit{a}) shows the output of the \texttt{ETAS} package's \texttt{rates} function \citep{jalilian2019etas}, which includes four subplots depicting the background seismicity rate, the total spatial seismicity rate, the clustering coefficient, and the conditional intensity function at the end of the study period.
        These subplots collectively provide insights into the temporal and spatial distribution of seismicity, underscoring regions and times of heightened earthquake likelihood.
        Panel (\textit{b}) illustrates the results of stochastic declustering \citep{zhuang2002stochastic} applied to the earthquake catalog (the slice from the year 2000).
        The left-hand side displays the geographical distribution and magnitude of background (immigrant) earthquakes.
        The right-hand side presents the location and magnitude of triggered (offspring/aftershock) earthquakes, which are directly influenced by preceding seismic activity.
    }
    \label{fig-etas}
\end{figure}

\subsection{Renewal Hawkes processes}\label{renewal-hawkes-processes}

From the immigration-birth representation (Section~\ref{sec-immigration-birth-view}) we can see that the Hawkes process is a \emph{Poisson cluster process}.
That is, the cluster locations (immigrants) arrive by a homogeneous Poisson process, while the clusters (births) are drawn as independent inhomogeneous Poisson processes.
The \emph{renewal Hawkes process} \citep{wheatley2016hawkes} is a natural extension of this idea, where the cluster locations are not drawn from a homogeneous Poisson process but instead from a renewal process (that is, a process with i.i.d. interarrival times which can be non-exponentially distributed).

In this setting, waiting times between immigrants are i.i.d. and are distributed according to some p.d.f. \(g(\cdot)\), which is related to the renewal intensity via
\[
\lambda(w) = \frac{g(w)}{1-\int_0^w g(s) \, \mathrm{d} s} \,.
\]
The conditional intensity function of the renewal Hawkes process is then
\[
\lambda^*_t = \lambda(t-T_{I[N(t)]}) + \sum_{T_i < t} \mu(t-T_i)\,,
\]
where
\[
I[N(t)] = \max\left\{j \in \{1,\dots,N(t)\} \,:\, t_j \in \{T_i^{(0)}\} \right\}
\]
is the index of the most recent immigrant event before \(t\).

The log-likelihood is straightforward to derive and has the same form as that of the standard Hawkes process (cf.~Equation~\ref{eq-log-likelihood}).
One can derive an EM algorithm to maximise this log-likelihood by introducing latent variables representing immigrant status and parent relationships.

The process has many interesting properties, including that the mean of the process \(M(t) = \mathbb{E}[N_t]\) uniquely satisfies the following integral equation \citep{chen2018direct}:
\[
M(t) = \int_0^t \left(1+K(t-s) + M(t-s) \right) \lambda(s) \exp\left(-\int_0^s \lambda(x) \, \mathrm{d} x\right) \mathrm{d} s,
\]
where \(K(t)\) is the expected number of events of the nonstationary process up to time \(t\), which in turn uniquely satisfies the following further integral equation:
\[
K(t) = \int_0^t \left(1 + K(t-s) \right) \mu(s) \, \mathrm{d} s\,.
\]
It is also known how to efficiently evaluate the likelihood for the renewal Hawkes process model \citep[see][]{chen2018direct}.

Much as with the classical Hawkes process model, one can extend the renewal Hawkes process model to the multivariate setting \citep{stindl2018likelihood}.
In this setting, one keeps track of \(M\) types of individuals, and one creates a process for each type, with the feature that immigrants are of the same type, but offspring can be created for any type according to a set of excitation functions.
One can again evaluate the likelihood for this model \citep{stindl2018likelihood}.

Simulation of the renewal Hawkes process \citep{chen2018direct} and its multivariate extension \citep{stindl2018likelihood} is based on their cluster process representations (see also Section~\ref{sec-immigration-birth-view}).

\subsection{Binned Data}\label{sec-binned-data}

When dealing with Hawkes processes, a common challenge arises --- the absence of specific timestamps for events.
Instead, we often only possess information about the frequency of events within a specified interval, while the exact timings remain undisclosed.

This phenomenon is known as \emph{binned data} and is illustrated in Figure~\ref{fig:time_series}.
Binned data is a phenomenon in which the value of an observation is only partially known.
In our case, the times at which events happen are \emph{censored}.

\begin{figure}[h!]
    (\textit{a}) \\
    \includegraphics[width=\columnwidth]{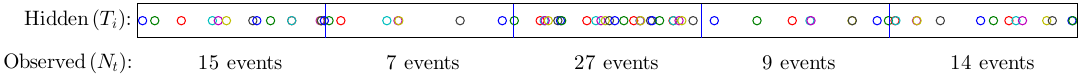}
    (\textit{b}) \\
    \includegraphics[width=\textwidth]{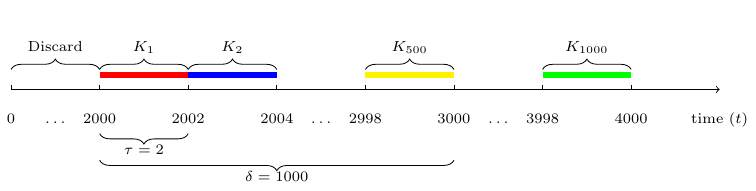}
    \caption{
            Panel (\textit{a}): The event times generated from a Hawkes process.
            The event times are \emph{not observable} (hidden).
            However, the aggregated number of events are \emph{observable}.
            There are 15, 7, 27, 9, and 14 events observed, as indicated.
            The specific times of occurrence for these events are censored, a phenomenon commonly referred to as binned data.
            Panel (\textit{b}): A timeline with \(T=4000\), \(\tau=2\), \(\delta = 1000\) thus giving \(n_\sharp=1000\) intervals and \(n_\ast := n_\sharp - \Delta = 500\) where \(\Delta = \lceil\delta/\tau\rceil =500\).
            For example \(K_1\) denotes the number of events falling in bin \(1\) of length \(\tau=2\).
        }
        \label{fig:time_series}
        \label{fig:line-grap}
\end{figure}

Despite recent studies in the methodology to draw inference in the context of binned data \citep{peto1973experimental, turnbull1976empirical, gentleman1994maximum, bohning1996interval, hu2009generalized}, relatively little work has attempted to model the number of events within a fixed interval using Hawkes processes.
Existing inference techniques on binned data almost always assume the tractability of the likelihood function.

A significant contribution to this field was initially proposed in 2010 by Aït Sahalia and colleagues and subsequently published in \citet{ait2015modeling}.
The work of \citet{dafonseca2014} builds upon this foundation, with the authors examining the functionals related to event counts over an interval for a Hawkes process.

The primary tool employed in drawing inference from binned data relies on the application of the Dynkin's formula within the framework of Hawkes processes.
The Markov property allows one to invoke certain tools to obtain the moments for the number of events in \emph{an interval}, rather than the number of events at the current time \(t > 0\).
Among these tools are the \emph{infinitesimal generator} and martingale techniques.
For a Markov process \(X_t\), consider a function \(f : D \rightarrow \mathbb{R}\).
The infinitesimal generator of the process, denoted \(\mathcal{A}\), is defined by
\[\mathcal{A}f(x) := \lim_{h\to 0} \frac{\mathbb{E}\left[f(X_{t+h})|X_t=x\right]-f(x)}{h} .\]

For every function \(f\) in the domain of the infinitesimal generator, the process
\[
    M_t := f(X_t) - f(X_0) - \int_0^t \mathcal{A} f(X_u) \,\mathrm{d}u
\]
is a martingale \citep{oksendal2007applied}.
Thus, for \(t > s\) we have
\[
    \mathbb{E}\left[ f(X_t) \!-\! \int_0^s \mathcal{A} f(X_u) \,\mathrm{d}u \,\Big|\, \mathcal{H}_s \right] = f(X_s) \!-\! \int_0^s \mathcal{A}f(X_u) \mathrm{d}u
\]
by the martingale property of \(M\).
Rearranging, we finally obtain Dynkin's formula \citep[see][]{medvegyev:2009}:
\begin{equation}\phantomsection\label{eq-dynkin}{
    \mathbb{E}\left[ f(X_t) \,|\, \mathcal{H}_s \right] = f(X_s) + \mathbb{E}\left[ \int_s^t \mathcal{A}f(X_u) \mathrm{d}u \,\Big|\, \mathcal{H}_s \right].
}\end{equation}

Following the concept of moment matching \citep{hall2004generalized}, the primary objective is to compute the theoretical moments and \textit{align} them with the empirical moments in the context of Hawkes processes.
We briefly outline these steps here, following the presentation of \citet{dafonseca2014} and \citet{ait2015modeling} and utilizing the exponentially decaying Hawkes process defined in Definition~\ref{def-exp-hawkes-process}, reiterated below:

\begin{equation}\phantomsection\label{eq-new-intensity}{
    \lambda^{\ast}_t = \lambda + (\lambda_0 - \lambda) \, \mathrm{e}^{-\beta t} + \sum_{T_i<t} \alpha \, \mathrm{e}^{-\beta (t-T_i)}.
}\end{equation}

Note that, according to Equation~\ref{eq-new-intensity}, the impact of the second term on \(\lambda^{\ast}_t\) diminishes exponentially as time passes.

\textbf{Theoretical moments.}
With these expressions at our disposal, we derive the following three quantities which will be useful in the inference for Hawkes processes.

\begin{proposition}[]\protect\hypertarget{prp-expectation-main}{}\label{prp-expectation-main}

The long term expectation of the number of jumps over an interval \(\tau\) is given by
\begin{equation}\phantomsection\label{eq-m1}{
    w^{\tau}_1 := \lim_{t\to\infty} \mathbb{E}\left[N_{t+\tau}-N_t \,|\, \lambda_0\,\right]=\frac{\lambda\beta}{\beta-\alpha}\tau.
}\end{equation}

\end{proposition}

\begin{proposition}[]\protect\hypertarget{prp-variance-main}{}\label{prp-variance-main}

The long term variance of the number of jumps over an interval \(\tau\) is given by
\begin{equation}\phantomsection\label{eq-m2}{
\begin{aligned}
    w^{\tau}_2 
    :=& \lim_{t\to\infty} \mathbb{E}\left[(N_{t+\tau}-N_t)^2 \,|\, \lambda_0\,\right]-
(\mathbb{E}\left[N_{t+\tau}-N_t \,|\, \lambda_0\,\right])^2 \\
    =& \frac{\lambda\beta}{\beta-\alpha}\left(\tau\frac{\beta^2}{(\beta-\alpha)^2}+\left(1-\frac{\beta^2}{(\beta-\alpha)^2}\right)\frac{1-\mathrm{e}^{-(\beta-\alpha)\tau}}{\beta-\alpha}\right).
\end{aligned}
}\end{equation}

\end{proposition}

\begin{proposition}[]\protect\hypertarget{prp-covariance-main}{}\label{prp-covariance-main}

The long term covariance of the number of jumps over an interval \(\tau\) and lag \(\delta\) is given by
\begin{equation}\phantomsection\label{eq-m3}{
\begin{aligned}
    w^{\tau}_3
    :=& \lim_{t\to\infty} \mathbb{E}\left[(N({t+\tau})-N(t))(N(t+2\tau+\delta)-N(t+\tau+\delta)) \,|\, \lambda_0\,\right] \\
      & \quad-\mathbb{E}\left[(N({t+\tau})-N(t))\,|\, \lambda_0\,\right]\mathbb{E}\left[(N(t+2\tau+\delta) \,|\, \lambda_0\,\right] \\
    =&  \frac{\lambda\beta\alpha(2\beta-\alpha)( \mathrm{e}^{-(\beta-\alpha)\tau}-1)^2}{2(\beta-\alpha)^4} \mathrm{e}^{-(\beta-\alpha)\delta}.
\end{aligned}
}\end{equation}

\end{proposition}

\begin{proof}
The proofs of Propositions~\ref{prp-expectation-main}, \ref{prp-variance-main}, and \ref{prp-covariance-main} are a simple application of Dynkin's formula as in Equation~\ref{eq-dynkin}.
\end{proof}

\textbf{Empirical moments.}
We begin by representing \(\mu_1\), \(\mu_2\), and \(\mu_3\) as the empirical mean, variance, and covariance of the number of events within an interval.
The specifications for \(\tau\) and \(\delta\) are predetermined and will be provided in the forthcoming Equations~\ref{eq-emp1}, \ref{eq-emp2} and \ref{eq-emp3} respectively.

The quantities \(w^{\tau}_1\), \(w^{\tau}_2\), and \(w^{\tau}_3\) given in Equations~\ref{eq-m1}, \ref{eq-m2} and \ref{eq-m3} respectively are functions of \(\vartheta=(\lambda,\beta,\alpha)\) which we want to estimate.
Let \(K_i\) denote the number of events falling in bin~\(i\) after discarding a pre-determined number of events.
Also, let \(n_{\sharp}\) be the number of bins remaining.
In this case, we have
\begin{align}
    m_1 &= \frac{1}{n_{\sharp}}\sum_{i=1}^{n_{\sharp}}K_i \label{eq-emp1}\\
    m_2 &= \frac{1}{n_{\sharp}}\sum_{i=1}^{n_{\sharp}}K_i^2-m_1^2 \label{eq-emp2}\\
    m_3 &= \frac{1}{n_\ast}\sum_{i=1}^{n_\ast} \Big(K_i\times K_{i+\Delta}\Big) -  \left(\frac{1}{n_\ast}\sum_{i=1}^{n_\ast} K_i\right)\times\left(\frac{1}{n_\ast} \sum_{i=1}^{n_\ast}K_{i+\Delta}\right), \label{eq-emp3}
\end{align}
where \(n_\ast := n_{\sharp} - \Delta\) and \(\Delta = \lceil\delta/\tau\rceil\).
Figure~\ref{fig:line-grap} gives an illustration of the empirical part of the moment matching strategy.
It is customary to discard a number of events at the beginning of our collected since we are working under asymptotic expressions, wherein \(t\rightarrow\infty\).
We may then define a function \(\widehat{h}(\vartheta)\) of the form
\begin{align}
    \widehat{h}(\vartheta)=\begin{pmatrix}
    m_1 - w^{\tau}_1\\
    m_2 - w^{\tau}_2\\
    m_3 - w^{\tau}_3
    \end{pmatrix}.
\end{align}
In our case, there exists a unique solution to \(\widehat{h}(\vartheta)=0\).
The method of moments estimator \(\widehat{\vartheta}\) can be identified by minimizing the criterion function which is equivalent to the squared sum \(\widehat{h}(\vartheta)^{\prime} \widehat{h}(\vartheta)\) \citep[cf.][]{hall2004generalized}.

Extensions to multivariate Hawkes processes are explored in \citet{achab2018uncovering}, establishing a connection between the branching structure of a multivariate Hawkes process and its implications for Granger causality.
Their nonparametric cumulant matching method surpasses previous efforts, as seen in Wiener--Hopf algorithms in \citet{bacry2016first}, due to its diminished complexity.

For kernels that induce non-Markovian Hawkes processes, readers are directed to the works of \citet{Cui_Hawkes_Yi_2020} and \citet{cui2022moments}, which provide alternatives to methodologies presented in this section.

\subsubsection{Additional methodologies}\label{additional-methodologies}

A drawback of the binning estimation methods mentioned earlier lies in the necessity for a pre-determined selection of the grid size.
This prerequisite frequently leads to models that are either overfitted or underfitted.
To tackle this issue, \citet{donnet2020nonparametric} presents a Bayesian nonparametric approach to model Hawkes processes, eliminating the requirement for a pre-defined grid size.
They establish posterior concentration rates for HPs and demonstrate these results by employing a nonparametric histogram representation of the triggering kernel.
This kernel is specifically inspired by neuroscience, particularly in simulating the interactions of neurons in the brain.
In their formulation of the histogram kernel, the compact set \((0,\bar{E})\) is defined with \(\bar{K}\) components, change points at \(\textbf{s}=(s_0=0,s_1, \dots, s_{\bar{K}-1}, s_{\bar{K}}=E)\) and corresponding heights \(v=(v_1,\dots, v_{\bar{K}})\), subject to the constraints
\[
    \sum_{k=1}^{\bar{K}} v_k=1
\]
and
\[
    \phi(\,\cdot\,|\,\bar{K},v,\textbf{s}) = \delta \cdot \sum_{k=1}^{\bar{K}} \frac{v_k}{s_k - s_{k-1}} {\textbf{1}}_{(s_{k-1}, s_k)}(\cdot).
\]
The indicator variable \(\delta\) adheres to a Bernoulli distribution with some probability \(p\), determining whether the histogram function is active or if all component heights are zero.
The model parameters are inferred using Reversible-jump Markov chain Monte Carlo (RJMCMC) \citep{green1995reversible}, a method known for its computational demands.
RJMCMC is a multidimensional Bayesian inference technique that facilitates transitions between various parameter spaces determined by possible values of \(\bar{K}\).
This investigation reveals a drawback: RJMCMC is computationally expensive and exhibits slower mixing compared to standard approaches.

In the context of addressing missing data within scenarios characterized by gaps in recording windows for continuous time-points, as examined by \citet{le2018multivariate}, pertinent references can be found in the works cited therein \citet{leigh2022parameter}.
This issue is especially pertinent in situations where data is characterized by precise yet intermittent recordings.

\subsubsection{Discrete-time Hawkes processes}\label{sec-discrete-time-hawkes}

Instead of trying to fit a continuous-time model to discrete-time data, we may use the \emph{discrete-time Hawkes process}, which adapts the classical Hawkes process to this setting.
This model has been applied to diverse datasets, from modelling terrorist activity \citep{porter2012self, white2013a} to COVID-19 cases \citep{browning2021simple}.
We will first adjust the counting process terminology from continuous time to discrete time then define the discrete-time Hawkes process.

In the \(t \in \mathbb{N}_0\) discrete-time case \(N_t\) still has the interpretation of the number of events up to time \(t\), though we can now have multiple arrivals at the same time.
Specifically, say that \(Y_t := N_t - N_{t-1} \in \mathbb{N}_0\) is the number of arrivals at time \(t\), and \(N_t = \sum_{s=1}^t Y_s\).
Instead of the conditional intensity process, we now overload the meaning of \(\lambda^*_t\) to refer to
\begin{equation}\phantomsection\label{eq-discrete-time-conditional-intensity}{
    \lambda^*_t = \mathbb{E}[ N_t - N_{t-1} \mid \mathcal{H}_{t-1}] = \mathbb{E}[ Y_t \mid \mathcal{H}_{t-1}] \,.
}\end{equation}
The history \(\mathcal{H}_t\) is now the information available up to (and including!) time \(t\), which is just \(Y_1, \dots, Y_t\).
The \(Y_t\) random variable typically has, when conditioned on \(\mathcal{H}_{t-1}\), a fixed parametric distribution such as the negative binomial distribution.
Denote \(p(\cdot ; m, \psi)\) to be the p.m.f. for this parametric distribution with mean~\(m\) and dispersion \(\psi\), then from Equation~\ref{eq-discrete-time-conditional-intensity} \(m = \lambda^*_t\).

The discrete-time Hawkes process is then the discrete-time counting process \(N_t\) whose \(\lambda^*_t\) for \(t \in \mathbb{N}_0\) is (analogous to Equation~\ref{eq-hawkes-intensity})
\begin{equation}\phantomsection\label{eq-discrete-time-hawkes-intensity}{
    \lambda^*_t = \lambda + \sum_{i=1}^{t-1} Y_i \, \mu(t-i) \,,
}\end{equation}
where the sum is interpreted as zero if \(t=1\).
Algorithm~\ref{alg:discrete-simulation} simply shows the discrete-time Hawkes process simulation algorithm which is implied by Equation~\ref{eq-discrete-time-hawkes-intensity}.

\begin{multicols}{2}

We still have the interpretations of \(\lambda > 0\) as the base intensity and \(\mu : \mathbb{N} \to \mathbb{R}_+\) as the excitation function.
It is common to set \(\mu(s) = \eta g(s)\), where \(g\) is a probability mass function with support over \(\mathbb{N}\) (not \(\mathbb{N}_0\)) and \(\eta > 0\) represents the expected number of offspring per arrival (just as in Section~\ref{sec-immigration-birth-view}).

\columnbreak

\begin{algorithm}[H]
\caption{Simulate discrete-time Hawkes}
\label{alg:discrete-simulation}
\begin{algorithmic}
\Require $\theta = (\lambda, \mu, \psi)$, $T$
\State $\lambda^*_1 \gets \lambda$
\State $Y_1 \sim p(\cdot ; \lambda^*_1, \psi)$
\For{$t \gets 2$ to $T$}
    \State $\lambda^*_t \gets \lambda + \sum_{i=1}^{t-1} Y_i \, \mu(t-i)$
    \State $Y_t \sim p(\cdot ; \lambda^*_t, \psi)$
\EndFor
\State \Return $\{Y_1, \dots, Y_T\}$
\end{algorithmic}
\end{algorithm}

\end{multicols}

\subsection{Dynamic contagion process}\label{sec-dynamic-contagion-process}

The inception of point processes that encompass both self and external excitations can be traced back to 2002 when the model was introduced by the authors \citet{BremaudMassoulie2002} under highly general conditions.
Subsequently, \citet{DassiosZhao2011} integrated the shot-noise Cox processes as the external component, while retaining the Hawkes process for self-excitations.
This combined model is referred to as the \emph{dynamic contagion process}.
The dynamic contagion process has a stochastic intensity function \(\lambda(t)\) of the form:
\[
    \lambda^*_t \!=\! \mathcal{B}_0(t) + \sum_{i:t>T_i} \mathcal{B}_1(Y_i,t\!-\!T_i) \!+\! \sum_{i:t>S_i} \mathcal{B}_2(X_i,t\!-\!S_i),
\]
where \(\mathcal{B}_0\), \(\mathcal{B}_1\) and \(\mathcal{B}_2\) are functions whose definitions will be made precise in Definition~\ref{def-main-intensity} below.
The sequence \((T_i)_{i\geq 1}\) denotes the event times of \(N_t\), where the occurrence of an event induces the intensity to grow by an amount \(\mathcal{B}_1(Y_i,t\!-\!T_i)\): this element captures self-excitation.
At the same time, external events can occur at times \(S_i\) and stimulate with a portion of \(\mathcal{B}_2(X_i,t\!-\!T_i)\): this is the externally-excited part.
The quantities \(X\) and \(Y\) are positive random variables describing the amplitudes by which \(\lambda\) increases during event times.
The quantity \(\mathcal{B}_0:\mathbb{R}_+\mapsto\mathbb{R}_+\) denotes the \emph{deterministic} base intensity.

We now specify the forms for \(\mathcal{B}_0\), \(\mathcal{B}_1\), and \(\mathcal{B}_2\) based on the dynamic contagion process defined in \citet{DassiosZhao2011}.

\begin{definition}[]\protect\hypertarget{def-main-intensity}{}\label{def-main-intensity}

The dynamic contagion process is a point process \(N_t\) on \(\mathbb{R}_+\) with the non-negative \(\mathcal{F}_t\) conditional random intensity
\begin{equation}\phantomsection\label{eq-detailed-intensity-function}{
    \lambda^*_t \!=\! a \!+\! (\lambda_0\!-\!a) \mathrm{e}^{-\delta t} \!+\! \sum_{i=1}^{N_t} Y_i \mathrm{e}^{-\delta (t \!-\! T_i)} \!+\! \sum_{i=1}^{J_t} X_i \mathrm{e}^{-\delta (t \!-\! S_i)},
}\end{equation}
for \(t\geq 0\), where we have the following features:

\begin{itemize}
\tightlist
\item
  \textbf{Deterministic background.} \(a\geq 0\) is the constant mean-reverting level, \(\lambda_0>a\) is the initial intensity at time \(t=0\), \(\delta>0\) is the constant rate of exponential decay.
  \(\mathcal{B}_0(t)=a+(\lambda_0-a) \mathrm{e}^{-\delta t}\);
\item
  \textbf{External-excitations.} \(X_i\) are levels of excitation from an external factor.
  They form a sequence of independent and identically distributed positive elements with distribution function \(H(c)\), \(c>0\).
  \(S_i\) are the times at which external events happen and it follows a Poisson process \(J_t\) of constant rate \(\rho>0\).
  Note that \(\mathcal{B}_2(X_i,t-S_i) := X_i \mathrm{e}^{-\delta(t-S_i)}\).
\item
  \textbf{Self-excitations.} \(Y_i\) are levels of \emph{self-excitation}, a sequence of independent and identically distributed positive elements with distribution function \(G(h)\), \(h>0\), occurring at random \(T_i\).
  Following the occurrence of events, the impact of these events will saturate and the rate at which this occurs is determined by the constant \(\delta\).
  Note that \(\mathcal{B}_1(Y_i,t-T_i) := Y_i \mathrm{e}^{-\delta(t-T_i)}\).
\end{itemize}

\end{definition}

Note that from Definition~\ref{def-main-intensity}, if we set \(X \equiv 0\) and \(Y\) to be a constant, we retrieve the original Hawkes model.
If we set \(Y\equiv 0\) and \(X\) to be a positive random elements, we get the model proposed by \citet{CoxIsham1980}.
In addition, setting \(X \equiv Y \equiv 0\) returns us the inhomogeneous Poisson process.
Furthermore, letting \(X \equiv Y \equiv 0\) and \(\lambda_0 = a\) simplify to the Poisson process.
This is a deliberate choice with \(\delta\) being shared between the background rate \(\mathcal{B}_0\) \(\mathcal{B}_1\) and \(\mathcal{B}_2\) to ensure that the process inherits the Markov property, see \citet{BremaudMassoulie2002}.
This property is essential to the derivation of the closed form functionals such as the mean and variance of the point process.
The term \(\delta\) determines the rate at which the process decays exponentially from following arrivals of self-excited and externally-excited events.

Inference for point processes with the combined elements of endogenous and exogenous components have been studied in the machine learning community.
\citet{slinderman} introduced a multidimensional counting process combining a sparse log Gaussian Cox process \citep{moller-lgcp} and a Hawkes component to uncover latent networks in the point process data.
They showed how the superposition theorem of point processes enables the formulation of a fully Bayesian inference algorithm.
\citet{mohler-hawkes-cox} considered a self-exciting process with background rate driven by a log Gaussian Cox process and performed inference based on an efficient Metropolis adjusted Langevin algorithm for filtering the intensity.
\citet{simma} proposed an expectation-maximization inference algorithm for Cox-type processes incorporating self-excitations via marked point processes and applied to data from a very large social network.
We note that a variety of methods have been developed for the estimation of self \emph{or} external excitations for point processes: variational flavors \citep{mangion-keyuan}, expectation propagation \citep{cseke-marginals}, the usage of thinning points and uniformization for non-stationary renewal processes \citep{gunter, teh-rao}.

\subsection{Stochastic differential equations and beyond}\label{stochastic-differential-equations-and-beyond}

This section provides a concise overview of several additional avenues in Hawkes processes research: stochastic differential equations (SDEs), graphs, and neural networks.
We will explore the applications of SDEs in Hawkes processes in some detail, focusing particularly on their formulation and inference.

In the first direction that concerns SDEs, \citet{lee2016hawkes} generalized both the classical model and that of \citet{DassiosZhao2011} to incorporate randomness in the triggering kernel.
They introduced contagion parameters to regulate the levels of excitation.
In this model, each tier of the excitation function is treated as a stochastic process, addressed through an SDE whose parameters are inferred using Bayesian methods.
The objective of this model is to provide an improved approximation, particularly in scenarios where the intensities of self-excitation are heightened with correlated levels of contagion.

The model is similar to that of Equation~\ref{eq-detailed-intensity-function}:
\begin{equation}\phantomsection\label{eq-first-lambda}{
    \lambda^*_t = a + (\lambda_0-a) \mathrm{e}^{-\delta t} + \sum_{i:t>T_i} Y(T_i) \, \mathrm{e}^{-\delta(t-T_i)}.
}\end{equation}
However, this time \(Y\) is a stochastic process.
We let \(Y_i := Y(T_i)\) for ease of notation.
The excitation levels \(Y\) measure the impact of clustering or contagion of the event times.
To see this, observe in Equation~\ref{eq-first-lambda} that whenever \(Y\) is high and positive, it imposes a greater value on the intensity \(\lambda\), thereby increasing the probability of generating an event in a shorter period of time.
This phenomenon causes the clustering observed.

SDEs are utilized to depict the progression of excitations' levels.
Expressing the development of contagiousness in mathematical terms involves formulating an equation that includes a derivative, as elaborated further below.
The alterations in excitation levels are presumed to adhere to the stochastic differential equation:
\[
    Y_t = \int_0^t \hat{\mu}(t, Y_s) \,\mathrm{d}s + \int_0^t \hat{\sigma}(s, Y_s) \,\mathrm{d}W_s,
\]
where \(W\) is a standard Brownian motion and \(t\in[0,T]\).
Different settings of the functionals \(\hat{\mu}\) and \(\hat{\sigma}\) lead to different versions of SDEs.
An important criterion for selecting the appropriate choices of the pair \((\hat{\mu},\hat{\sigma})\) essentially comes down to how we decide to model the levels of excitation.
A standing assumption is that the contagion process must be positive, i.e.,
\begin{equation}\phantomsection\label{eq-y-positive-assumption}{
\textrm{The contagion parameters $Y_t > 0$, $\forall t \geq 0$}.
}\end{equation}
This is necessary because the levels of excitation, denoted by \(Y\), act as parameters that scale the magnitude of the influence of each past event, taking into account the fact that the conditional intensity function is non-negative.
Some notable examples include Geometric Brownian Motion: \(\hat{\mu}=\mu+\frac{1}{2}\sigma^2 Y\), \(\hat{\sigma}=\sigma Y\) \citep{kloeden-platen}; Square-Root-Processes: \(\hat{\mu}=k(\mu-Y)\), \(\hat{\sigma}=\sigma\sqrt{Y}\); Langevin equation: \(\hat{\mu}=k(\mu-Y)\), \(\hat{\sigma}=\sigma\), and their variants \citep{liptser-shiryaev}.

While the positivity of \(Y\) is guaranteed for Geometric Brownian Motion (GBM), this may not be true for other candidates such as Langevin dynamics or Square-Root-Processes.
This is because they possess the inherent property that nothing prevents them from going negative, and thus they may not be suitable choices to model the levels of excitation.
Specifically, Square-Root-Processes can be negative if the so-called Feller condition \(2k\mu>\sigma^2\) is not satisfied \citep{feller-condition, liptser-shiryaev}.
In some real-life applications, this condition may not be respected, thus raising the possibility of violating the assumption in Equation~\ref{eq-y-positive-assumption}.

For this reason, we focus on the Geometric Brownian Motion and we tilt the Langevin dynamics by exponentiating it so that the positivity of \(Y\) is ensured \citep{black-kara}.

\begin{itemize}
\tightlist
\item
  Geometric Brownian Motion (GBM):
  \[
    Y_t = \int_0^t \left(\mu+\frac{1}{2}\sigma^2\right) Y_s \,\mathrm{d}s + \int_0^t \sigma Y_s\, \mathrm{d}B_s,
  \]
  where \(\mu\in\mathbb{R}\) and \(\sigma>0\).
\item
  Exponential Langevin:
  \[
    \log Y_t = \int_0^t k(\mu-Y_s) \,\mathrm{d}s + \int_0^t \sigma \,\mathrm{d}B_s,
  \]
  where \(k,\mu\in\mathbb{R}\) and \(\sigma>0\).
\end{itemize}

In light of Section~\ref{sec-immigration-birth-view}, the generative view provided by Equation~\ref{eq-first-lambda} facilitates a systematic treatment of scenarios where each observed event time can be distinguished into \emph{immigrants} and \emph{offspring}, terms we will define shortly.
This distinction is crucial for conducting inference algorithms.
An event time \(T_i\) is termed an \emph{immigrant} if it is generated from the background intensity \(a + Y_0 \mathrm{e}^{-\delta t}\); otherwise, we refer to it as an \emph{offspring}.
Hence, it is natural to introduce a variable describing the specific process to which each event time \(T_i\) corresponds.
We do that by introducing the random variable \(Z_i := Z_{ij}\), where
\[
\begin{aligned}
    Z_{i0} &= 1\quad\text{if event $i$ is an immigrant, }\\
    Z_{ij} &= 1\quad\text{if event $i$ is an offspring of $j$. }
\end{aligned}
\]
An equivalent interpretation of the set \(\{Z_{ij}=1\}\) is the following: event \(i\) was caused by event \(j\).
Furthermore, each \(Z_i\) is a special indicator matrix where only one of its element is unity, i.e.~the vector \(Z_i=(Z_{i0},Z_{i1}, \dots, Z_{i,(i-1)})\) contains a single 1 and 0 otherwise.
For a fixed \(i\), we have \(\sum_{j=0}^{i-1} Z_{ij} = 1\).

\textbf{Inference.} A hybrid MCMC algorithm is introduced, which updates the parameters individually, either by direct draws using Gibbs sampling or via the Metropolis-Hastings (MH) algorithm.
This hybrid approach merges the characteristics of both the Gibbs sampler and the MH algorithm, offering substantial flexibility in designing inference methods for various parameters within our model.
To understand the mechanics of this concept, let's examine a two-dimensional parameterization as an example.
Assuming the posterior \(\p(\Theta_B|\Theta_A)\) follows a known distribution, we can directly conduct inference using the Gibbs sampler.
However, if \(\p(\Theta_A|\Theta_B)\) can only be evaluated but not directly sampled, we resort to employing an MH algorithm to update \(\Theta_A\) given \(\Theta_B\).
In the MH step, candidates are drawn from \(\q(\Theta^{\prime}_A|\Theta^{(k)}_A,\Theta^{(k)}_B)\), indicating that the current step may depend on the past draw of \(\Theta_A\).
The Metropolis step samples from \(\q(\Theta_A|\Theta^{(k)}_A,\Theta^{(k)}_B)\), suggesting that we draw \(\Theta^{(k+1)}_A\sim \q(\Theta^{\prime}_A|\Theta^{(k)}_A,\Theta^{(k)}_B)\), and the acceptance or rejection of the proposed candidate is determined based on the acceptance probability denoted by \(AP\):
\begin{equation}\phantomsection\label{eq-ap}{
    AP(\Theta^{\prime}_A) = \min\left(\frac{\p(\Theta^{\prime}_A|\Theta_B)\q(\Theta^{(k)}_A|\Theta^{\prime}_A,\Theta^{(k)}_B)}{\p(\Theta_A|\Theta_B)\q(\Theta^{\prime}_A|\Theta^{(k)}_A,\Theta^{(k)}_B)}, 1\right).
}\end{equation}
The hybrid algorithm is as follows: given \(\left(\Theta^{(0)}_A,\Theta^{(0)}_B\right)\), for \(k=0, 1, \dots, K\),

\begin{enumerate}
\def\labelenumi{\arabic{enumi}.}
\tightlist
\item
  Sample \(\Theta_A^{(k+1)}\sim \q(\Theta_A|\Theta^{(k)}_A,\Theta^{(k)}_B)\) and \emph{accept} or \emph{reject} \(\Theta_A^{(k+1)}\) based on Equation~\ref{eq-ap},
\item
  Sample \(\Theta_B^{(k+1)}\sim \p(\Theta_B|\Theta^{(k+1)}_A)\).
\end{enumerate}

A hybrid algorithm typically comprises any combination of Gibbs and MH steps.
These algorithms are highly flexible, as the only essential requirement is that posterior distributions can be computed.
Gibbs sampling is utilized to perform inferences for \(Z\) and the parameters of the stochastic differential equations (both GBM and exponential Langevin).
For the inferences of the remaining parameters \(a\), \(\lambda_0\), \(\delta\), and \(Y_i\) as in Equation~\ref{eq-first-lambda}, Normal distributions may be employed as proposals.
For implementation details including the discretization of \(Y\), refer to \citet{lee2016hawkes}.

\subsubsection{Additional research directions}\label{additional-research-directions}

The second avenue of exploration centers around graph-based methodologies.
These approaches empower users to grasp the interconnections among components within multivariate hyperparameters by revealing the underlying network structure.
Typically, this is achieved by estimating the infectivity matrix, where the \mbox{\(ij\)-th} element describes the expected number of offspring events anticipated in dimension \(i\) given an event in dimension \(j\) \citep{liu2018exploiting}.
\citet{slinderman} integrated hyperparameters into random graph models, decomposing the infectivity matrix into a binary adjacency matrix indicating network sparsity, and a weight matrix characterizing interaction strength.
Model parameters were estimated using a parallelizable Gibbs sampler.
\citet{guo2015bayesian} expanded on this by introducing a novel Bayesian language model, aiming to reveal latent networks, particularly focusing on examining dialogue evolution within a social network over time.
Subsequently, \citet{linderman2017bayesian} addressed the task of inferring latent structures within a social network characterized by incomplete data, presenting a sequential Monte Carlo approach for data recovery.
Expanding on their prior research, \citet{bacry2015hawkes} improved their approach by incorporating sparsity and low-rank induced penalization.
This adjustment resulted in an excitation matrix characterized by a limited number of non-zero and independent rows, with the aim of enhancing scalability and refining kernel estimation.
Following a similar trajectory, \citet{bacry2016mean} and \citet{bacry2020sparse} pursued inference for higher-dimensional hyperparameters by representing the excitation function as a low-rank approximation with regularization.
To improve computational efficiency in parameter recovery across higher dimensions, the Mean-Field Hypothesis assumed minor fluctuations in stochastic intensity.

In a the next direction, the nonlinearity inherent in the intensity function is approached by incorporating a neural network.
Recurrent neural networks, for instance, encode sequences of input and output states, where each state is influenced by the preceding one, and the hidden state captures information from past states.
The parameters of the neural network are optimized through a fitting procedure involving nonlinear functions, such as sigmoidal or hyperbolic tangent.
Addressing challenges encountered in recurrent neural networks, the long short-term memory (LSTM) architecture resolves the vanishing gradient problem and extends memory by modeling HP intensities of multiple events.
This is accomplished through the utilization of `forget gates' that regulate the influence of past events on the current state \citep{mei2017neural}.
Other neural network approaches, such as self-attentive/transformer models \citep{zuo2020transformer, zhang2020self} and graph convolution networks \citep{shang2019geometric}, demonstrate promising computational efficiencies.

\clearpage

\section*{Supplementary Materials: Simulation Algorithms Pseudocode}

\setcounter{algorithm}{0}

\begin{algorithm}[H]
\caption{Hawkes Process Simulation by Thinning}
\label{alg:thinning-simulation}
\begin{algorithmic}
\Require $\theta = (\lambda, \mu)$, $T$ \Comment{Model parameters and simulation time horizon}
\State $n = 0$, $t \gets 0$, $\lambda^*_t \gets \lambda$
\While{True}
    \State $M \gets \lambda^*_t$ \Comment{Require that $M \ge \lambda^*_s$ for $t < s < T_{n+1}$}
    \State Generate $E$ as an exponential random variable with rate $M$
    \State $t \gets t + E$ \Comment{Update current time}
    \State Stop looping if the current time $t$ exceeds $T$
    \State Update $\lambda^*_t$ based on $\mu(t)$ and past events $\{T_1, \dots, T_n\}$
    \State Generate $u$ as a uniform random variable over $(0, M)$
    \State If $u > \lambda^*_t$ then skip to the next iteration
    \State $n \gets n + 1$ and $T_n \gets t$  \Comment{Record the time of the event}
    \State Update $\lambda^*_t$ to reflect the occurrence of the new event at time $t$
\EndWhile
\State \Return $\{T_1, \dots, T_n\}$
\end{algorithmic}
\end{algorithm}

\begin{algorithm}[H]
\caption{Hawkes Process Exact Simulation by Composition}
\label{alg:exact-simulation}
\begin{algorithmic}
\Require $\theta = (\lambda, \alpha, \beta)$, $N$ \Comment{Model parameters and number of events}
\State $t \gets 0$, $\lambda^*_t \gets \lambda$
\For{$i \gets 1$ to $N$}
    \State Generate $U_1, U_2$ as i.i.d. standard uniform random variables
    \State $E_1 \gets - \log(U_1) / \lambda$
    \State $E_2 \gets - \log(1 + \frac{\beta}{\lambda^*_t + \alpha - \lambda} \log(U_2)) / \beta$
    \State $t_{\text{prev}} \gets t$
    \State $t \gets t + \min(E_1, E_2)$ \Comment{Choose the next event time}
    \State $T_i \gets t$ \Comment{Store the time of the $i$-th arrival}
    \State $\lambda^*_t \gets \lambda + (\lambda^*_{t_{\text{prev}}} + \alpha - \lambda) \cdot \exp(-\beta \cdot (t - t_{\text{prev}}))$ \Comment{See Equation~\ref{eq-exp-markov-decay}}
\EndFor
\State \Return $\{T_1, \dots, T_N\}$
\end{algorithmic}
\end{algorithm}

\end{document}